\documentclass[acmsmall]{acmart}

\AtBeginDocument{%
  \providecommand\BibTeX{{%
    \normalfont B\kern-0.5em{\scshape i\kern-0.25em b}\kern-0.8em\TeX}}}

\setcopyright{acmcopyright}
\acmJournal{TALG}
\acmYear{2019}\acmVolume{15}\acmNumber{3}\acmArticle{38}\acmMonth{6}
\acmDOI{10.1145/3326169}

\usepackage{microtype}
\usepackage{amsmath}
\usepackage{bbm}
\usepackage{xspace}
\usepackage{paralist}

\newtheorem{fact}{Fact}[section]
\newtheorem{claim}{Claim}[section]
\newtheorem{remark}{Remark}[section]
\newenvironment{proofof}[1]{{\vspace*{5pt} \noindent\bf Proof of #1:  }}{\hfill\rule{2mm}{2mm}\vspace*{5pt}}

\newcommand{\obm}{\textsf{Online Bipartite Matching}\xspace}
\newcommand{\vwobm}{\textsf{Online Vertex-Weighted Bipartite Matching}\xspace}

\newcommand{\ranking}{\textsf{Ranking}\xspace}

\newcommand{\expect}[2]{\operatorname{\mathbf E}_{#1}\left[#2\right]}
\newcommand{\vect}[1]{\ensuremath{\vec{#1}}}
\newcommand{\vecy}{\vect{y}}

\newcommand{\onev}{\ensuremath{\mathbbm{1}}}

\begin{document}

\title{Online Vertex-Weighted Bipartite Matching: Beating $1-\frac{1}{e}$ with Random Arrivals}
\titlenote{A preliminary version of this paper appeared in ICALP 2018~\cite{icalp/HTWZ18}.
	This work is partially supported by grants from Hong Kong RGC under
	the contract 17202715 and HKU17202115E.}

\author{Zhiyi Huang}
\affiliation{%
	\institution{Department of Computer Science, The University of Hong Kong}
}
\email{zhiyi@cs.hku.hk}

\author{Zhihao Gavin Tang}
\authornote{Part of the work was done when the author was a PhD student at the University of Hong Kong.}
\affiliation{%
	\institution{ITCS, Shanghai University of Finance and Economics, China}
}
\email{tang.zhihao@mail.shufe.edu.cn}

\author{Xiaowei Wu}
\authornote{Part of the work was done when the author was a postdoc at the University of Hong Kong. The research leading to these results has received funding from the European Research Council under the European Community's Seventh Framework Programme (FP7/2007-2013) / ERC grant agreement No. 340506.}
\affiliation{%
	\institution{Faculty of Computer Science, University of Vienna, Austria}
}
\email{wxw0711@gmail.com}

\author{Yuhao Zhang}
\affiliation{%
	\institution{Department of Computer Science, The University of Hong Kong, China}
}
\email{yhzhang2@cs.hku.hk}

\renewcommand{\shortauthors}{Z. Huang, Z. Tang, X. Wu, and Y. Zhang}

\begin{abstract}
	We introduce a weighted version of the ranking algorithm by Karp et al. (STOC 1990), and prove a competitive ratio of 0.6534 for the vertex-weighted online bipartite matching problem when online vertices arrive in random order. Our result shows that random arrivals help beating the 1-1/e barrier even in the vertex-weighted case.
	We build on the randomized primal-dual framework by Devanur et al. (SODA 2013) and design a two dimensional gain sharing function, which depends not only on the rank of the offline vertex, but also on the arrival time of the online vertex. To our knowledge, this is the first competitive ratio strictly larger than 1-1/e for an online bipartite matching problem achieved under the randomized primal-dual framework. Our algorithm has a natural interpretation that offline vertices offer a larger portion of their weights to the online vertices as time increase, and each online vertex matches the neighbor with the highest offer at its arrival. 
\end{abstract}



\maketitle

\section{Introduction}

With a wide range of applications, \obm and its variants are a focal point in the online algorithms literature.
Consider a bipartite graph $G(U, V, E)$ on vertices $U\cup V$, where the set $V$ of offline vertices is known in advance and vertices in $U$ arrive online.
On the arrival of an online vertex, its incident edges are revealed and the algorithm must irrevocably either match it to one of its unmatched neighbors or leave it unmatched.
In a seminal paper, Karp et al.~\cite{stoc/KarpVV90} proposed the \ranking algorithm, which picks at the beginning a random permutation over the offline vertices $V$, and matches each online vertex to the first unmatched neighbor according to the permutation.
They proved a tight competitive ratio $1-\frac{1}{e}$ of \ranking, when online vertices arrive in an arbitrary order.
The analysis has been simplified in a series of subsequent works~\cite{soda/GoelM08,sigact/BenjaminC08,soda/DevanurJK13}.
Further, the \ranking algorithm has been extended to other variants of the \obm problem, including the vertex-weighted case~\cite{soda/AggarwalGKM11}, the random arrival model~\cite{stoc/KarandeMT11,stoc/MahdianY11}, and the Adwords problem~\cite{jacm/MehtaSVV07,esa/BuchbinderJN07,stoc/DevanurJ12}.

As a natural generalization, \vwobm was considered by Aggarwal et al.~\cite{soda/AggarwalGKM11}.
In this problem, each offline vertex $v\in V$ has a non-negative weight $w_v$, and the objective is to maximize the total weight of the matched offline vertices.
A weighted version of the \ranking algorithm was proposed in~\cite{soda/AggarwalGKM11} and shown to be $(1-\frac{1}{e})$-competitive, matching the problem hardness in the unweighted version.
They fix a non-increasing perturbation function $\psi: [0,1]\rightarrow [0,1]$, and draw a rank $y_v\in[0,1]$ uniformly and independently for each offline vertex $v\in V$.
The offline vertices are then sorted in decreasing order of the \emph{perturbed weight} $w_v\cdot \psi(y_v)$.
Each online vertex matches the first unmatched neighbor on the list upon its arrival.
It is shown that by choosing the perturbation function $\psi(y) := 1-e^{y-1}$, the weighted \ranking algorithm achieves a tight competitive ratio $1-\frac{1}{e}$.
In a subsequent work, Devanur et al.~\cite{soda/DevanurJK13} simplified the analysis under the randomized primal-dual framework and gave an alternative interpretation of the algorithm: each offline vertex $v$ makes an offer of value $w_v\cdot (1-g(y_v))$ as long as it is not matched, where $g(y) := e^{y-1} = 1-\psi(y)$, and each online vertex matches the neighbor that offers the highest.

Motivated by the practical importance of \obm and its applications for online advertisements, another line of research seeks for a better theoretical bound beyond the worst-case hardness result provided by Karp et al.~\cite{stoc/KarpVV90}.
\obm problem with random arrivals was considered independently by Karande et al.~\cite{stoc/KarandeMT11} and Mahdian et al.~\cite{stoc/MahdianY11}.
They both studied the performance of \ranking assuming that online vertices arrive in a uniform random order and proved competitive ratios $0.653$ and $0.696$ respectively.
On the negative side, Karande et al.~\cite{stoc/KarandeMT11} explicitly constructed an instance for which \ranking performs no better than $0.727$, which is later improved to $0.724$ by Chan et al.~\cite{soda/ChanCWZ14}.
In terms of problem hardness, Manshadi et al.~\cite{mor/ManshadiGS12} showed that no algorithm can achieve a competitive ratio larger than $0.823$.

The natural next step is then to consider the \vwobm problem with random arrivals.
\emph{Do random arrivals help beating $1-\frac{1}{e}$ even in the vertex-weighted case?}

\begin{table}[h]
	\centering
	\begin{tabular}{c|c|c}
		& Arbitrary Arrivals & Random Arrivals \\
		\hline
		Unweighted & $1-\frac{1}{e} \approx 0.632$~\cite{stoc/KarpVV90, sigact/BenjaminC08, soda/DevanurJK13, soda/GoelM08} & $0.696$~\cite{stoc/MahdianY11} \\
		\hline
		Vertex-weighted & $1-\frac{1}{e} \approx 0.632$~\cite{soda/AggarwalGKM11, soda/DevanurJK13} & \textbf{$\bf 0.6534$ (this paper)}
	\end{tabular}
\end{table}

\subsection{Our Results and Techniques}

We answer this affirmatively by showing that a generalized version of the \ranking algorithm achieves a competitive ratio $0.6534$.

\begin{theorem}
	There exists a $0.6534$-competitive algorithm for the \vwobm problem with random arrivals.
\end{theorem}

Interestingly, we do not obtain our result by generalizing existing works that break the $1-\frac{1}{e}$ barrier on the unweighted case~\cite{stoc/KarandeMT11,stoc/MahdianY11} to the vertex-weighted case. Instead, we take a totally different path, and build our analysis on the randomized primal-dual technique introduced by Devanur et al.~\cite{soda/DevanurJK13}, which was used to provide a more unified analysis of the algorithms for the \obm problem with arbitrary arrival order and its extensions.

We first briefly review the proof of Devanur et al.~\cite{soda/DevanurJK13}. The randomized primal-dual technique can be viewed as a charging argument for sharing the gain of each matched edge between its two endpoints.
Recall that in the algorithm of~\cite{soda/AggarwalGKM11,soda/DevanurJK13}, each unmatched offline vertex offers a value of $w_v\cdot(1-g(y_v))$ to online vertices, and each online vertex matches the neighbor that offers the highest at its arrival.
Whenever an edge $(u, v)$ is added to the matching, where $v\in V$ is an offline vertex and $u\in U$ is an online vertex, imagine a total gain of $w_v$ being shared between $u$ and $v$ such that $u$ gets $w_v\cdot(1-g(y_v))$ and $v$ gets $w_v\cdot g(y_v)$.
Since $g$ is non-decreasing, the smaller the rank of $v$, the smaller share it gets.
For any edge $(u,v)$ and any fixed ranks of online vertices other than $v$, they showed that by fixing $g(y) = e^{y-1}$, the expected gains of $u$ and $v$ (from all of their incident edges) combined is at least $(1-\frac{1}{e})\cdot w_v$ over the randomness of $y_v$.
This implies the $1-\frac{1}{e}$ competitive ratio.

Now we consider the problem with random arrivals.

Analogous to the offline vertices, as the online vertices arrive in random order, in the gain sharing process,
it is natural to give an online vertex $u$ a smaller share if $u$ arrives early (as it is more likely be get matched), and a larger share when $u$ arrives late.
Thus we consider the following version of the weighted \ranking algorithm.

Let $y_u$ be the arrival time of online vertex $u\in U$, which is chosen uniformly at random from $[0,1]$.
Analogous to the ranks of the offline vertices, we also call $y_u$ the rank of $u\in U$. 
Fix a function $g: [0,1]^2\rightarrow [0,1]$ that is non-decreasing in the first dimension and non-increasing in the second dimension.
On the arrival of $u\in U$, each unmatched neighbor $v\in V$ of $u$ makes an offer of value $w_v\cdot (1-g(y_v,y_u))$, and $u$ matches the neighbor with the highest offer.
This algorithm straightforwardly leads to a gain sharing rule for dual assignments: whenever $u\in U$ matches $v\in V$, let the gain of $u$ be $w_v\cdot(1-g(y_v,y_u))$ and the gain of $v$ be $w_v\cdot g(y_v,y_u)$.
It suffices to show that, for an appropriate function $g$, the expected gain of $u$ and $v$ combined is at least $0.6534 \cdot w_v$ over the randomness of both $y_u$ and $y_v$.

The main difficulty of the analysis is to give a good characterization of the behavior of the algorithm when we vary the ranks of both $u\in U$ and $v\in V$, while fixing the ranks of all other vertices arbitrarily.
The previous analysis for the unweighted case with random arrivals~\cite{stoc/KarandeMT11,stoc/MahdianY11} heavily relies on a symmetry between the random ranks of offline vertices and online vertices: Properties developed for the offline vertices in previous work directly translate to their online counterparts.
Unfortunately, the online and offline sides are no longer symmetric in the vertex-weighted case.
In particular, for the offline vertex $v$, an important property is that for any given rank $y_u$ of the online vertex $u$, we can define a unique marginal rank $\theta$ such that $v$ will be matched if and only if its rank $y_v< \theta$.
However, it is not possible to define such a marginal rank for the online vertex $u$ in the vertex-weighted case: As its arrival time changes, its matching status may change back and forth.
In particular, since the function $g$ depends on the arrival time of $u$, it may happen that $u$ prefers neighbor $v$ to $z$ at one arrival time, but prefers $z$ to $v$ at another.
The most important technical ingredient of our analysis is an appropriate lower bound on the expected gain which allows us to partially characterize the worst-case scenario (in the sense of minimizing the lower bound on the expected gain).
Further, the worst-case scenario does admit simple marginal ranks even for the online vertex $u$.
This allows us to design a symmetric gain sharing function $g$ and complete the competitive analysis of $0.6534$.

As we will discuss in Section~\ref{sec:conclusion}, our framework may be able to give stronger lower bound on the competitive ratio, potentially matching or even improving the one of Mahdian and Yan~\cite{stoc/MahdianY11}, if we had a tight analysis of a complex system of differential inequalities.
Numerical results suggest that the integration shown in Section~\ref{sec:conclusion} may give a much larger lower bound on the competitive ratio than the one we present in this paper.
However, giving a tight analysis on the integration is highly non-trivial.
Indeed, a significant portion of our analysis, e.g., Section~\ref{sec:improved} and part of Section~\ref{sec:simple}, is devoted to provide analyzable relaxations on this integration.

\subsection{Other Related Works}

There is a vast literature on problems related to \obm. For space reasons, we only list some of the most related here.

Kesselheim et al.~\cite{esa/KesselheimRTV13} considered the edge-weighted \obm problem with random arrivals, and proposed a $\frac{1}{e}$-competitive algorithm.
The competitive ratio is tight as it matches the lower bound on the classical secretary problem~\cite{mor/BuchbinderJS14}.
Wang and Wong~\cite{icalp/WangW15} considered a different model of \obm problem with both sides of vertices arriving online (in an arbitrary order):
A vertex can only actively match other vertices at its arrival; if it fails to match at its arrival, it may still get matched passively by other vertices later.
They showed a $0.526$-competitive algorithm for a fractional version of the problem.

Recently, Cohen and Wajc~\cite{soda/CohenW18} considered the \obm (with arbitrary arrival order) on regular graphs, and provided a $(1-O(\sqrt{\log d/d}))$-competitive algorithm, where $d$ is the degree of vertices.
Very recently, Huang et al.~\cite{stoc/HKTWZZ18} proposed a fully online matching model, in which all vertices of the graph arrive online (in an arbitrary order).
Extending the randomized primal-dual technique, they obtained competitive ratios above $0.5$ for both bipartite graphs and general graphs.

Similar but different from the \obm problem with random arrivals, in the stochastic \obm, the online vertices arrive according to some known probability distribution (with repetition).
Competitive ratios breaking the $1-\frac{1}{e}$ barrier have been achieved for the unweighted case~\cite{focs/FeldmanMMM09,esa/BahmaniK10,esa/BrubachSSX16} and the vertex-weighted case~\cite{wine/HaeuplerMZ11,mor/JailletL14,esa/BrubachSSX16}.

The \obm problem with random arrivals is closely related to the oblivious matching problem~\cite{rsa/Aronson1995,soda/ChanCWZ14,esa/AbolhassaniCCEH16} (on bipartite graphs).
It can be easily shown that \ranking has equivalent performance on the two problems.
Thus competitive ratios above $1-\frac{1}{e}$~\cite{stoc/KarandeMT11,stoc/MahdianY11} directly translate to the oblivious matching problem.
Generalizations of the problem to arbitrary graphs have also been considered, and competitive ratios above half are achieved for the unweighted case~\cite{rsa/Aronson1995,soda/ChanCWZ14} and vertex-weighted case~\cite{esa/AbolhassaniCCEH16,talg/ChanCW18}.

\section{Preliminaries}

We consider the \vwobm with random arrival order. Let $G(U, V,E)$ be the underlying graph, where vertices in $V$ are given in advance and vertices in $U$ arrive online in random order.
Each offline vertex $v\in V$ is associated with a non-negative weight $w_v$.
Without loss of generality, we assume the arrival time $y_u$ of each online vertex $u \in U$ is drawn independently and uniformly from $[0,1]$.
Mahdian and Yan~\cite{stoc/MahdianY11} use another interpretation for the random arrival model. They denote the order of arrival of online vertices by a permutation $\pi$ and assume that $\pi$ is drawn uniformly at random from the permutation group $S_n$. It is easy to see the equivalence between two interpretations:
The algorithm draws $n$ independent random variables from $[0,1]$ uniformly at random before any online vertex arrives, and assigns the $i$-th smallest variable to the $i$-th online vertex in the random permutation as its arrival time.

\paragraph{Weighted \ranking}
Fix a function $g: [0,1]^2\rightarrow [0,1]$ such that $\frac{\partial g(x,y)}{\partial x} \geq 0$ and $\frac{\partial g(x,y)}{\partial y} \leq 0$.
Each offline vertex $v\in V$ draws independently a random rank $y_v\in[0,1]$ uniformly at random. 
Upon the arrival of online vertex $u\in U$, $u$ is matched to its unmatched neighbor $v$ with maximum $w_v\cdot(1-g(y_v,y_u))$.

\begin{remark}
	In the adversarial model, Aggarwal et al.'s algorithm~\cite{soda/AggarwalGKM11} can be interpreted as choosing $g(y_v, y_u) := e^{y_v-1}$ in our algorithm.
	Our algorithm is a direct generalization of theirs to the random arrival model.
\end{remark}

For simplicity, for each $u\in U$, we also call its arrival time $y_u$ the rank of $u$. We use $\vecy : U\cup V \rightarrow [0,1]$ to denote the vector of all ranks.

Consider the linear program relaxation of the bipartite matching problem and its dual.
\begin{align*}
	\max: \quad & \textstyle \sum_{(u,v)\in E} w_v\cdot x_{uv} && \qquad\qquad & \min: \quad & \textstyle\sum_{u \in U} \alpha_u + \sum_{v \in V} \alpha_v\\
	\text{s.t.} \quad & \textstyle \sum_{v:(u,v)\in E} x_{uv} \leq 1 && \forall u\in U & \text{s.t.} \quad & \alpha_u + \alpha_v \geq w_v && \forall (u,v)\in E \\
	& \textstyle \sum_{u:(u,v)\in E} x_{uv} \leq 1 && \forall v\in V & & \alpha_u \geq 0 && \forall u \in U \\
	& x_{uv} \geq 0 && \forall (u,v)\in E & & \alpha_v \geq 0 && \forall v \in V
\end{align*}

\paragraph{Randomized Primal-Dual}
Our analysis builds on the randomized primal-dual technique by Devanur et al.~\cite{soda/DevanurJK13}.
We set the primal variables according to the matching produced by \ranking, i.e. $x_{uv}=1$ if and only if $u$ is matched to $v$ by \ranking, and set the dual variables so that the dual objective equals the primal.
In particular, we split the gain $w_v$ of each matched edge $(u,v)$ between vertices $u$ and $v$; the dual variable for each vertex then equals the share it gets.
Given primal feasibility and equal objectives, the usual primal-dual techniques would further seek to show approximate dual feasibility, namely, $\alpha_u + \alpha_v \ge F\cdot w_v$ for every edge $(u, v)$, where $F$ is the target competitive ratio.
Observe that the above primal and dual assignments are themselves random variables.
Devanur et al.~\cite{soda/DevanurJK13} claimed that the primal-dual argument goes through given approximate dual feasibility in expectation. We formulate this insight in the following lemma and include a proof for completeness.

\begin{lemma}\label{lemma:dual_fitting}
	\ranking is $F$-competitive if we can set (non-negative) dual variables such that 
	\begin{compactitem}
		\item $\sum_{(u,v)\in E} x_{uv} = \sum_{u \in V} \alpha_u$; and
		\item $\expect{\vecy}{\alpha_u+\alpha_v} \geq F\cdot w_v$ for all $(u,v)\in E$.
	\end{compactitem}
\end{lemma}
\begin{proof}
	We can set a feasible dual solution $\tilde{\alpha}_u := \expect{\vecy}{\alpha_u}/F$ for all $u\in V$. It's feasible because we have $\tilde{\alpha}_u + \tilde{\alpha}_v = \expect{\vecy}{\alpha_u + \alpha_v}/F \ge w_v$ for all $(u,v) \in E$.
	Then by duality we know that the dual solution is at least the optimal primal solution \textsf{PRIMAL}, which is also at least the optimal offline solution of the problem: 
	$\sum_{u \in V} \tilde{\alpha}_u \ge \textsf{PRIMAL} \ge \textsf{OPT}$.
	Then by the first assumption, we have
	$\textsf{OPT} \leq \sum_{u \in V} \tilde{\alpha}_u = \sum_{u \in V} \frac{\expect{\vecy}{{\alpha_u}}}{F} = \frac{1}{F}\expect{\vecy}{\sum_{u \in V} {\alpha_u}} =  \frac{1}{F}\expect{\vecy}{\sum_{(u,v)\in E} w_v\cdot x_{uv}} = \frac{1}{F}\expect{}{\textsf{ALG}}$,
	which implies an $F$ competitive ratio.
\end{proof}

In the rest of the paper, we set
\[
g(x,y) = \frac{1}{2}\big(h(x)+1-h(y)\big), \qquad \forall x,y\in[0,1]
\]
where $h: [0,1]\rightarrow [0,1]$ is a non-decreasing function (to be fixed later) with $h'(x) \leq h(x)$ for all $x\in[0,1]$.
Observe that $\frac{\partial g(x,y)}{\partial x} = \frac{1}{2}h'(x) \ge 0$ and $\frac{\partial g(x,y)}{\partial y} = -\frac{1}{2}h'(y) \le 0$. By definition of $g$, we have $g(x,y) + g(y,x) = 1$. 
Moreover, for any $x,y\in[0,1]$, we have the following fact that will be useful for our analysis.
\begin{claim}
	\label{eq:partial_g}
	$\frac{\partial g(x,y)}{\partial y} \ge g(x,y)-1$.
\end{claim}
\begin{proof}
	$
	\frac{\partial g(x,y)}{\partial y} = -\frac{1}{2}h'(y) \geq -\frac{1}{2}h(y) \geq \frac{1}{2}(h(x)+1-h(y))-1 = g(x,y)-1. 
	$
\end{proof}

\section{A Simple Lower Bound} \label{sec:simple}

In this section, we prove a slightly smaller competitive ratio, $\frac{5}{4}-e^{-0.5}\approx 0.6434$, as a warm-up of the later analysis.

We reinterpret our algorithm as follows.
As time $t$ increases, each unmatched offline vertex $v \in V$ is dynamically priced at $w_v \cdot g(y_v, t)$. Since $g$ is non-increasing in the second dimension, the prices do not increase as time increases. 
Upon the arrival of $u \in U$, $u$ can choose from its unmatched neighbors by paying the corresponding price.
The utility of $u$ derived by choosing $v$ equals $w_v-w_v\cdot g(y_v,y_u)$. Then $u$ chooses the one that gives the highest utility.
Recall that $g$ is non-decreasing in the first dimension.
Thus, $u$ prefers offline vertices with smaller ranks, as they offer lower prices.

This leads to the following monotonicity property as in previous works~\cite{soda/AggarwalGKM11,soda/DevanurJK13}.

\begin{fact}[Monotonicity]\label{fact:monotonicity}
	For any $\vecy$, if $v\in V$ is unmatched when $u\in U$ arrives, then when $y_v$ increases, $v$ remains unmatched when $u$ arrives.
	Equivalently, if $v\in V$ is matched when $u\in U$ arrives, then when $y_v$ decreases, $v$ remains matched when $u$ arrives.
\end{fact}

\paragraph{Gain Sharing}
The above interpretation induces a straightforward gain sharing rule: whenever $u \in U$ is matched to $v \in V$, let $\alpha_v := w_v \cdot g(y_v, y_u)$ and $\alpha_u:= w_v \cdot (1-g(y_v, y_u)) = w_v \cdot g(y_u, y_v)$.

\medskip

Note that the gain of an offline vertex is larger if it is matched earlier, i.e., being matched earlier is more beneficial for offline vertices ($\alpha_v$ is larger).
However, the fact does not hold for online vertices.
For each online vertex $u\in U$, the earlier $u$ arrives (smaller $y_u$ is), the more offers $u$ sees.
On the other hand, the prices of offline vertices are higher when $u$ comes earlier.
Thus, it is not guaranteed that earlier arrival time $y_u$ induces larger $\alpha_u$.

This is where our algorithm deviates from previous ones~\cite{soda/AggarwalGKM11,soda/DevanurJK13}, in which the prices of offline vertices are static (independent of time).
The above observation is crucial and necessary for breaking the $1-\frac{1}{e}$ barrier in the random arrival model.

To apply Lemma~\ref{lemma:dual_fitting}, we consider a pair of neighbors $v\in V$ and $u \in U$. We fix an arbitrary assignment of ranks to all vertices but $u,v$.
Our goal is to establish a lower bound of $\frac{1}{w_v}\cdot \expect{}{\alpha_u + \alpha_v}$, where the expectation is simultaneously taken over $y_u$ and $y_v$.	

\begin{lemma}\label{lemma:thresholds}
	For each $y\in [0,1]$, there exist thresholds $1 \geq \theta(y) \geq \beta(y) \geq 0$ such that when $u$ arrives at time $y_u = y$,
	\begin{compactitem}
		\item if $y_v < \beta(y)$, $v$ is matched when $u$ arrives;
		\item if $y_v \in (\beta(y), \theta(y))$, $v$ is matched to $u$;
		\item if $y_v > \theta(y)$, $v$ is unmatched after $u$'s arrival.
	\end{compactitem}
	Moreover, $\beta(y)$ is a non-decreasing function and if $\theta(x) = 1$ for some $x\in[0,1]$, then $\theta(x') = 1$ for all $x' \geq x$.
\end{lemma}
\begin{proof}
	Consider the moment when $u$ arrives.
	By Fact~\ref{fact:monotonicity}, there exists a threshold $\beta(y_u)$ such that $v$ is matched when $u$ arrives iff $y_v < \beta(y_u)$.
	Now suppose $y_v > \beta(y_u)$, in which case $v$ is not matched when $u$ arrives.
	Thus $v$ is priced at $w_v\cdot g(y_v,y_u)$ and $u$ can get utility $w_v\cdot g(y_u,y_v)$ by choosing $v$.
	
	Recall that $g(y_u,y_v)$ is non-increasing in terms of $y_v$.
	Let $\theta(y_u) \geq \beta(y_u)$ be the minimum value of $y_v$ such that $v$ is not chosen by $u$.
	In other words, when $\beta(y_u) < y_v < \theta(y_u)$, $u$ matches $v$
	and when $y_v > \theta(y_u)$, $v$ is unmatched after $u$'s arrival.
	
	Next we show that $\beta$ is a non-decreasing function of $y_u$.
	By definition, if $y_v < \beta(y_u)$, then $v$ is matched when $u$ arrives. Straightforwardly, when $y_u$ increases to $y'_u$ (arrives even later), $v$ would remain matched.
	Hence, we have $\beta(y'_u) \geq \beta(y_u)$ for all $y'_u > y_u$, i.e. $\beta$ is non-decreasing (refer to Figure~\ref{fig:theta_beta}).
	
	\begin{figure}[htb]
		\centering
		\centering\includegraphics[width = 0.4\textwidth]{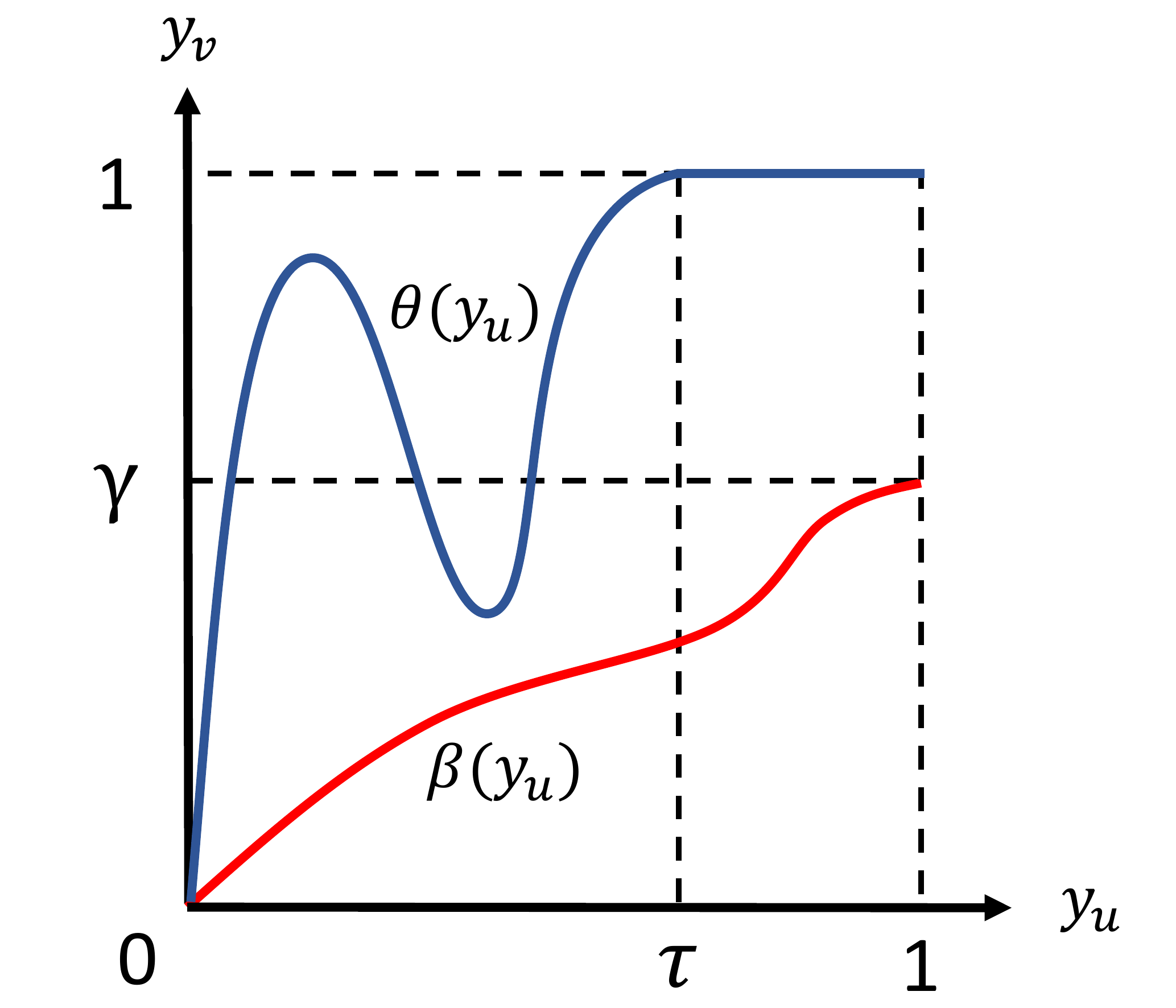}
		\centering\includegraphics[width = 0.4\textwidth]{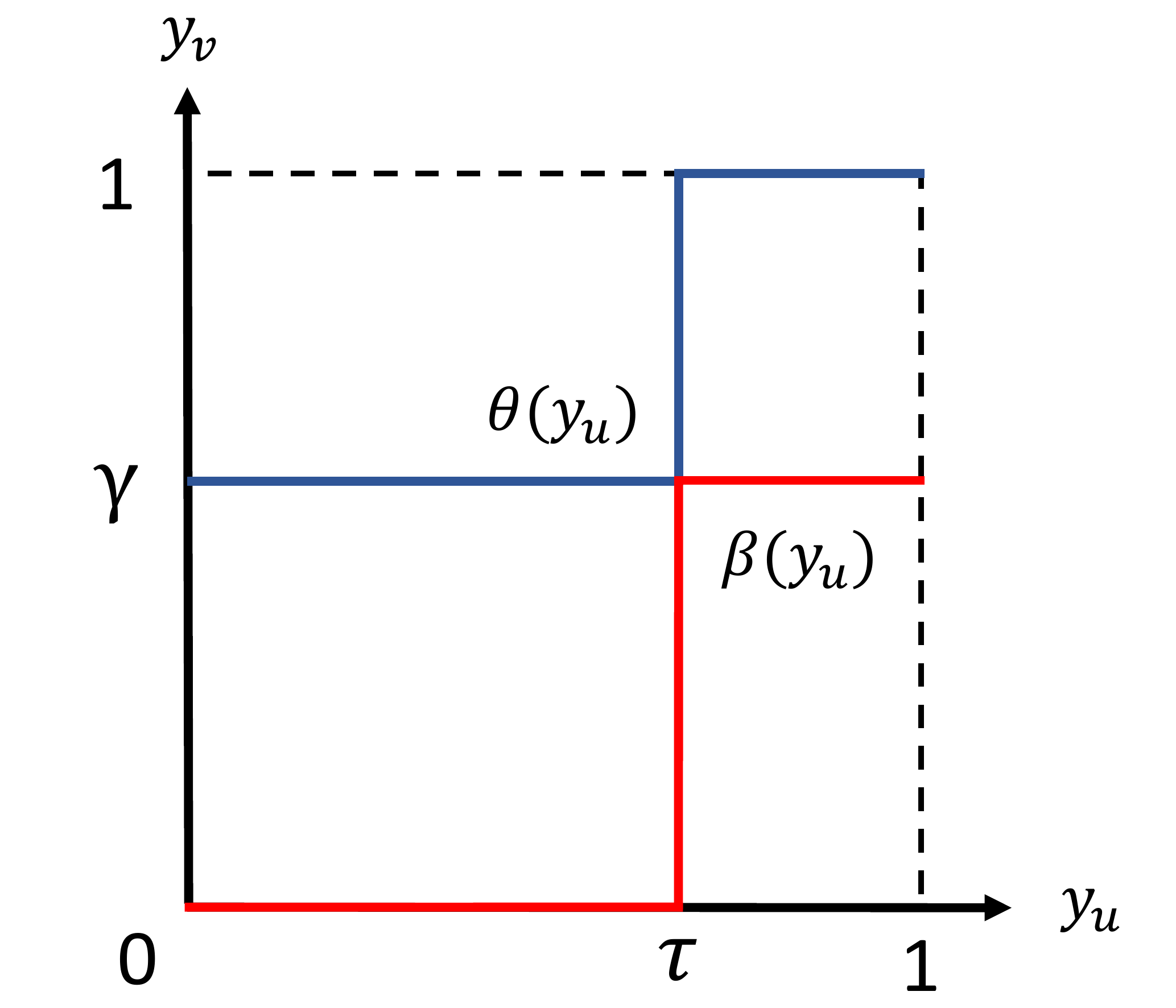}
		\caption{$\theta(y_u)$ and $\beta(y_u)$ (left hand side); truncated $\theta(y_u)$ and $\beta(y_u)$ (right hand side).}
		\label{fig:theta_beta}
	\end{figure}
	
	Finally, we show that if $\theta(x) = 1$ for some $x\in[0,1]$, then $\theta(x') = 1$ for all $x'\geq x$.
	Assume for the sake of contradiction that $\theta(x') < 1$ for some $x' > x$.
	In other words, when $y_u = x'$ and $y_v = 1$, $v$ is unmatched when $u$ arrives, but $u$ chooses some vertex $z\neq v$, such that $w_z\cdot g(x', y_z) > w_v\cdot g(x', 1)$. 
	
	Now consider the case when $u$ arrives at time $y_u = x$. Recall that we have $\theta(x) = 1$, which means that $u$ matches $v$ when $y_u = x$ and $y_v = 1$.
	By our assumption, both $v$ and $z$ are unmatched when $u$ arrives at time $x'$.
	Thus when $u$ arrives at an earlier time $x$, both $v$ and $z$ are unmatched.
	Moreover, choosing $z$ induces utility
	\begin{align*}
		w_z\cdot g(x, y_z) & = w_z\cdot g(x', y_z)\cdot \frac{g(x,y_z)}{g(x',y_z)} > w_v\cdot g(x', 1)\cdot \frac{g(x,y_z)}{g(x',y_z)} \\
		& = w_v\cdot g(x', 1)\cdot \frac{h(x)+1-h(y_z)}{h(x')+1-h(y_z)}\geq w_v\cdot g(x', 1)\cdot \frac{h(x)+1-h(1)}{h(x')+1-h(1)} \\
		& = w_v\cdot g(x', 1)\cdot\frac{g(x,1)}{g(x',1)} = w_v\cdot g(x,1),
	\end{align*}
	where the second inequality holds since $h$ is a non-decreasing function and $x < x'$.
	
	This gives a contradiction, since when $y_u = x$ and $y_v = 1$, $u$ chooses $v$, while choosing $z$ gives strictly higher utility.	
\end{proof}

\begin{remark}
	In the previous analysis by Devanur et al.~\cite{soda/DevanurJK13} on the arbitrary arrival model, a single marginal rank (independent of $y_u$) of $v$ is defined, and they do not distinguish whether $v$ is matched with $u$, as the gain sharing depends only on the rank of $v$, e.g., the definition of $\beta$ is unnecessary.
\end{remark}

\begin{remark}
	Observe that the function $\theta$ is not necessarily monotone.
	This comes from the fact that $u$ may prefer $v$ to $z$ when $u$ arrives at time $t$ but prefer $z$ to $v$ when $u$ arrives later at time $t' > t$.
	Note that this happens only when the offline vertices have general weights: for the unweighted case, it is easy to show that $\theta$ must be non-decreasing.
\end{remark}

We define $\tau,\gamma\in[0,1]$, which depend on the input instance, as follows.

If $\theta(y) < 1$ for all $y\in[0,1]$, then let $\tau = 1$; otherwise let $\tau$ be the minimum value such that $\theta(\tau) = 1$.
Let $\gamma := \beta(1)$.
Note that it is possible that $\gamma\in\{0,1\}$.

Since $\beta$ is non-decreasing, we define $\beta^{-1}(x) := \sup\{ y: \beta(y)=x \}$ for all $x \leq \gamma$.	

In the following, we establish a lower bound for $\frac{1}{w_v}\cdot\expect{}{\alpha_u + \alpha_v}$.
\begin{lemma}[Main Lemma]\label{lemma:bipartite_random}
	For each pair of neighbors $u\in U$ and $v\in V$, we have
	\begin{equation*}
		\frac{1}{w_v}\cdot \expect{}{\alpha_u + \alpha_v} \ge \min_{0\leq \gamma, \tau\leq1} \Big\{(1-\tau)\cdot(1-\gamma) + \int_0^\gamma g(x, \tau)dx + \int_0^\tau g(x,\gamma)dx  \Big\}.
	\end{equation*}
\end{lemma}

It is worthwhile to make a comparison with a similar claim in the previous analysis by Devanur et al.~\cite{soda/DevanurJK13} on the arbitrary arrival model: $\frac{1}{w_v}\cdot \expect{}{\alpha_u + \alpha_v} \ge\min_\theta \{\int_0^\theta g(y) dy + 1-g(\theta) \}$.
The first term in their lower bound comes from the gain of the offline vertex $v$ while the $1-g(\theta)$ term comes from the fact that the online vertex $u$ has gained at least $1-g(\theta)$ for all values of $y_v$.
Compared to theirs, our lower bound beats $1-\frac{1}{e}$ by utilizing the trade-off between the gain $\int_0^\gamma g(x,\tau) dx$ of $v$ and the ``marginal'' arrival time $\tau$ of $u$: in the previous analysis, only the trade-off between the gain $1-g(\theta)$ of $u$ and the marginal rank $\theta$ of $v$ is utilized.

We prove Lemma~\ref{lemma:bipartite_random} by the following three lemmas.

Observe that for any $y_u\in [0, 1]$, if $y_v \in (\beta(y_u), \theta(y_u))$, $u,v$ are matched to each other, which implies $\alpha_u+\alpha_v = w_v$. Hence we have the following lemma immediately.

\begin{lemma}[Corner Gain]\label{lemma:corner_gain}
	$\expect{}{(\alpha_u+\alpha_v)\cdot\onev(y_u>\tau, y_v>\gamma)} = w_v\cdot (1-\tau)\cdot(1-\gamma)$.
\end{lemma}

Now we give a lower bound for the gain of $v$ when $y_v < \gamma$, i.e., $\alpha_v\cdot \onev(y_v<\gamma)$, plus the gain of $u$ when $y_v < \gamma$ and $y_u > \tau$, i.e., $\alpha_u\cdot\onev(y_v<\gamma, y_u>\tau)$.
The key to prove the lemma is to show that for all $y_v < \gamma$, no matter when $u$ arrives, we always have $\alpha_v\geq w_v\cdot g(y_v,\beta^{-1}(y_v))$.

\begin{lemma}[$v$'s Gain]\label{lemma:v_gain}
	$\expect{}{\alpha_v\cdot \onev(y_v<\gamma)+\alpha_u\cdot\onev(y_v<\gamma, y_u>\tau)} \geq w_v\cdot \int_0^\gamma g(x, \tau)dx$.
\end{lemma}
\begin{proof}
	Fix $y_v = x < \gamma$. 
	We first show that for all $y_u\in[0,1]$, $\alpha_v \geq w_v\cdot g(x, \beta^{-1}(x))$.	
	By definition, we have $\beta^{-1}(x) < 1$.
	Hence when $y_u > \beta^{-1}(x)$, $v$ is already matched when $u$ arrives.
	Suppose $v$ is matched to some $z\in U$, then we have $y_z \leq \beta^{-1}(x)$ and hence $\alpha_v \geq w_v\cdot g(x, \beta^{-1}(x))$.
	Now consider when $u$ arrives at time $y < \beta^{-1}(x)$.
	If $y > y_z$, then $v$ is still matched to $z$ when $u$ arrives, and $\alpha_v \geq w_v\cdot g(x, \beta^{-1}(x))$ holds. 
	Now suppose $y < y_z$. 
	We compare the two processes, namely when $y_u > \beta^{-1}(x)$ and when $y_u = y$.
	
	We show that for each vertex $w \in V$, the time it is matched is not later in the second case (compared to the first case).
	In other words, we show that decreasing the rank of any online vertex is not harmful for all offline vertices.
	Suppose otherwise, let $w$ be the first vertex in $V$ that is matched later when $y_u=y$ than when $y_u > \beta^{-1}(x)$.
	I.e. among all these vertices, $w$'s matched neighbor arrives the earliest when $y_u > \beta^{-1}(x)$.
	
	Let $u_1$ be the vertex $w$ is matched to when $y_u > \beta^{-1}(x)$ and $u_2$ be the vertex $w$ is matched to when $y_u = y$.
	By assumption, we have $y_{u_2} > y_{u_1}$. 
	Consider when $y_u=y$ and the moment when $u_1$ arrives, $w$ remains unmatched but is not chosen by $u_1$.
	However, $w$ is the first vertex that is matched later than it was when $y_u > \beta^{-1}(x)$, we know that at $u_1$'s arrival, the set of unmatched neighbor of $u_1$ is a subset of that when $y_u > \beta^{-1}(x)$.
	This leads to a contradiction, since $w$ gives the highest utility, but is not chosen by $u_1$.
	
	In particular, this property holds for vertex $v$, i.e. $v$ is matched earlier or at the arrival of $z$ and hence $\alpha_v \geq w_v\cdot g(x, y_z)\ge w_v\cdot g(x, \beta^{-1}(x))$, as claimed.
	
	Observe that for $y_v<\gamma$ and $y_u \in (\tau, \beta^{-1}(y_v))$, we have $\alpha_u + \alpha_v = w_v$.
	Thus for $y_v = x<\gamma$, we lower bound $\frac{1}{w_v}\cdot \expect{y_u}{\alpha_v\cdot \onev(y_v<\gamma)+\alpha_u\cdot\onev(y_v<\gamma, y_u>\tau)}$ by
	\begin{equation*}
		f(x, \beta^{-1}(x)) := g(x, \beta^{-1}(x))+\max\{0,\beta^{-1}(x)-\tau\}\cdot (1-g(x,\beta^{-1}(x))).
	\end{equation*}
	
	It suffices to show that $f(x, \beta^{-1}(x)) \geq g(x, \tau)$. Consider the following two cases.
	\begin{enumerate}	
		\item If $\beta^{-1}(x) < \tau$, then $f(x,\beta^{-1}(x)) = g(x,\beta^{-1}(x)) \geq g(x,\tau)$, since $\frac{\partial g(x,y)}{\partial y} \leq 0$.
		
		\item If $\beta^{-1}(x) \ge \tau$, then $f(x, \beta^{-1}(x))$ is non-decreasing in the second dimension, since
		\begin{equation*}
			\frac{\partial f(x,\beta^{-1}(x))}{\partial \beta^{-1}(x)} = \frac{\partial g(x,\beta^{-1}(x))}{\partial \beta^{-1}(x)} + 1-g(x,\beta^{-1}(x))
			- (\beta^{-1}(x)-\tau)\cdot \frac{\partial g(x,\beta^{-1}(x))}{\partial \beta^{-1}(x)} \geq 0,
		\end{equation*}
		where the inequality follows from Claim~\ref{eq:partial_g} and the fact that $\frac{\partial g(x,\beta^{-1}(x))}{\partial \beta^{-1}(x)} \leq 0$.
		Therefore, we have $f(x, \beta^{-1}(x)) \geq f(x,\tau) = g(x,\tau)$.
	\end{enumerate}
	
	Hence for every fixed $y_v = x <\gamma$ we have
	\begin{equation*}
		\expect{y_u}{\alpha_v\cdot \onev(y_v<\gamma)+\alpha_u\cdot\onev(y_v<\gamma, y_u>\tau)} \geq w_v\cdot g(x,\tau).
	\end{equation*}
	
	Taking integration over $x\in (0,\gamma)$ concludes the lemma.
\end{proof}

Next we give a lower bound for the gain of $u$ when $y_u < \tau$, i.e., $\alpha_u\cdot \onev(y_u<\tau)$,
plus the gain of $v$ when $y_u < \tau$ and $y_v > \gamma$, i.e., $\alpha_v\cdot\onev(y_u<\tau, y_v>\gamma)$.
The following proof is in the same spirit as in the proof of Lemma~\ref{lemma:v_gain}, although the ranks of offline vertices have different meaning from the ranks (arrival times) of online vertices.

Similar to the proof of Lemma~\ref{lemma:v_gain}, the key is to show that for all $y_u < \tau$, no matter what value $y_v$ is, the gain of $\alpha_u$ is always at least $w_v\cdot g(y_u,\theta(y_u))$.

\begin{lemma}[$u$'s Gain]\label{lemma:u_gain}
	$\expect{}{\alpha_u\cdot \onev(y_u<\tau)+\alpha_v\cdot\onev(y_u<\tau, y_v>\gamma)} \geq w_v\cdot \int_0^\tau g(x, \gamma)dx$.
\end{lemma}
\begin{proof}
	Fix $y_u = x < \tau$.
	By definition we have $\theta(x) < 1$.
	The analysis is similar to the previous.
	We first show that for all $y_v\in[0,1]$, we have $\alpha_u \geq w_v\cdot g(x, \theta(x))$.
	
	We use $\theta$ to denote the value that is arbitrarily close to, but larger than $\theta(x)$.
	By definition, when $y_v = \theta$, $u$ matches some vertex other than $v$.
	Thus we have $\alpha_u \geq w_v\cdot g(x, \theta(x))$.
	Hence, when $y_v > \theta$, i.e. $v$ has a higher price, $u$ would choose the same vertex as when $y_v = \theta$, and $\alpha_u \geq w_v\cdot g(x, \theta(x))$ still holds.
	
	Now consider the case when $y_v =y < \theta$.
	
	As in the analysis of Lemma~\ref{lemma:v_gain}, we compare two processes, when $y_v = \theta$ and when $y_v = y < \theta$. We show that for each vertex $w \in U$ (including $u$) with $y_{w}\leq x = y_u$, the utility of $w$ when $y_v = y$ is not worse than its utility when $y_v = \theta$.
	Suppose otherwise, let $w$ be such a vertex with earliest arrival time.
	
	Let $v'$ be the vertex that is matched to $w$ when $y_v = \theta$.
	Then we know that (when $y_v = y$) at $w$'s arrival, $w$ chooses a vertex that gives less utility comparing to $v'$. Hence, at this moment $v'$ is already matched to some $w'$ with $y_{w'} < y_{w}$.
	This implies that when $y_v = \theta$, $v'$ (which is matched to $w$) is unmatched when $w'$ arrives, but not chosen by $w'$.
	Therefore, $w'$ has lower utility when $y_v = y$ compared to the case when $y_v = \theta$, which contradicts the assumption that $w$ is the first such vertex.
	
	Observe that when $y_v \in (\gamma, \theta(x))$, we have $\alpha_u + \alpha_v = w_v$.	
	Thus for any fixed $y_u = x<\tau$, we lower bound $\frac{1}{w_v}\cdot\expect{y_v}{\alpha_u\cdot \onev(y_u<\tau)+\alpha_v\cdot\onev(y_u<\tau, y_v>\gamma)}$ by
	\begin{equation*}
		f(x, \theta(x)) := g(x, \theta(x))+\max\{0,\theta(x)-\gamma\}\cdot (1-g(x,\theta(x))).
	\end{equation*}
	
	In the following, we show that $f(x,\theta(x)) \geq g(x, \gamma)$. Consider the following two cases.
	\begin{enumerate}
		\item If $\theta(x) \leq \gamma$, then $f(x,\theta(x)) = g(x,\theta(x)) \geq g(x,\gamma)$, since $\frac{\partial g(x,y)}{\partial y} \leq 0$.
		
		\item If $\theta(x) > \gamma$, then
		\begin{equation*}
			\frac{\partial f(x,\theta(x))}{\partial \theta(x)} = \frac{\partial g(x,\theta(x))}{\partial \theta(x)} + 1-g(x,\theta(x))
			- (\theta(x)-\gamma)\cdot \frac{\partial g(x,\theta(x))}{\partial \theta(x)} \geq 0,
		\end{equation*}
		where the inequality follows from Claim~(\ref{eq:partial_g}) and $\frac{\partial g(x,\theta(x))}{\partial \theta(x)} \leq 0$.
		Therefore, we have $f(x,\theta(x)) \geq f(x,\gamma) = g(x,\gamma)$.	
	\end{enumerate}
	
	Finally, take integration over $x\in (0,\tau)$ concludes the lemma.
\end{proof}

\begin{proofof}{Lemma~\ref{lemma:bipartite_random}}
	Observe that
	\begin{align*}
		\alpha_u+\alpha_v = (\alpha_u+\alpha_v)\cdot\onev(y_u>\tau, y_v>\gamma) 
		& + \alpha_v\cdot \onev(y_v<\gamma)+\alpha_u\cdot\onev(y_v<\gamma, y_u>\tau) \\
		& + \alpha_u\cdot \onev(y_u<\tau)+\alpha_v\cdot\onev(y_u<\tau, y_v>\gamma)
	\end{align*}
	Combing Lemma~\ref{lemma:corner_gain},~\ref{lemma:v_gain} and~\ref{lemma:u_gain} finishes the proof immediately.
\end{proofof}

\begin{theorem}
	Fix $h(x) = \min\{1,e^{x-0.5}\}$.
	For any pair of neighbors $u$ and $v$, and any fixed ranks of vertices in $U\cup V\setminus\{u,v\}$, we have
	$\frac{1}{w_v}\cdot \expect{y_u,y_v}{\alpha_u + \alpha_v} \ge \frac{5}{4}-e^{-0.5} \approx 0.6434$.
\end{theorem}
\begin{proof}
	It suffices to show that the RHS of Lemma~\ref{lemma:bipartite_random} is at least $\frac{5}{4}-e^{-0.5}$. Since the expression is symmetric for $\tau$ and $\gamma$, we assume $\tau \geq \gamma$ without loss of generality.
	
	Let $f(\tau,\gamma)$ be the term on the RHS of Lemma~\ref{lemma:bipartite_random} to be minimized.
	By our choice of $g$,
	\begin{align*}
		f(\tau,\gamma) = & 1-\tau-\gamma+\tau\cdot\gamma + \frac{1}{2}\int_0^\gamma\big(h(x)+1-h(\tau)\big) dx + \frac{1}{2}\int_0^\tau\big(h(x)+1-h(\gamma)\big) dx \\
		= & 1-\frac{\tau}{2}(1+h(\gamma))-\frac{\gamma}{2}(1+h(\tau))+\tau\cdot\gamma + \frac{1}{2}\int_0^\gamma h(x) dx + \frac{1}{2}\int_0^\tau h(x) dx .
	\end{align*}
	
	Observe that
	\begin{equation*}
		\frac{\partial f(\tau,\gamma)}{\partial \tau} = \gamma -\frac{1}{2}(1+h(\gamma))-\frac{\gamma}{2}\cdot h'(\tau) + \frac{1}{2} h(\tau).
	\end{equation*}
	
	It is easy to check that
	\begin{equation*}
		\gamma -\frac{1}{2}h(\gamma)
		\begin{cases}
			\leq 0 \qquad \text{when } \gamma \leq \frac{1}{2},\\
			> 0 \qquad \text{when } \gamma > \frac{1}{2}.
		\end{cases}
	\end{equation*}
	
	Hence when $\gamma \leq \frac{1}{2}$, we have
	\begin{equation*}
		\frac{\partial f(\tau,\gamma)}{\partial \tau} \leq \gamma -\frac{1}{2}h(\gamma) - \frac{1}{2} (1-h(\tau)) \leq 0,
	\end{equation*}
	which means that the minimum is attained when $\tau = 1$.
	Note that when $\gamma \leq \frac{1}{2}$, we have 
	\begin{equation*}
		f(1,\gamma) = \frac{1}{2}(1-h(\gamma))+ \frac{1}{2}\int_0^\gamma h(x) dx + \frac{1}{2}\int_0^1 h(x) dx,
	\end{equation*}
	which attains its minimum at $\gamma = 0$ (since $h'(\gamma) = h(\gamma)$ for $\gamma \le \frac{1}{2}$):
	\begin{equation*}
		f(1,0) = \frac{1}{2}(1-e^{-0.5}) + \frac{1}{2}(\frac{1}{2}+1-e^{-0.5}) = \frac{5}{4} - e^{-0.5}\approx 0.6434.
	\end{equation*}
	
	When $\tau \geq \gamma > \frac{1}{2}$, we have
	\begin{equation*}
		\frac{\partial f(\tau,\gamma)}{\partial \tau} = \gamma -\frac{1}{2}(1+h(\gamma))-\frac{\gamma}{2}\cdot 0 + \frac{1}{2} = \gamma -\frac{1}{2}h(\gamma) > 0
	\end{equation*}
	
	Hence the minimum is attained when $\tau = \gamma$, which is
	\begin{equation*}
		f(\gamma,\gamma) = 1- 2\gamma+\gamma^2 + \int_0^\gamma h(x) dx.
	\end{equation*}
	
	Observe that
	\begin{equation*}
		\frac{d f(\gamma,\gamma)}{d \gamma} = -2+2\gamma+h(\gamma) \geq -2+1+1 = 0.
	\end{equation*}
	
	The minimum is attained when $\gamma = \frac{1}{2}$, which equals $f(\frac{1}{2}, \frac{1}{2}) = \frac{5}{4} - e^{-0.5} \approx 0.6434$.
\end{proof}

\section{Improving the Competitive Ratio} \label{sec:improved}

Observe that in Lemma~\ref{lemma:bipartite_random}, we relax the total gain of $\alpha_u+\alpha_v$ into two parts:
\begin{enumerate}
	\item when $y_u\ge\tau$ and $y_v\ge\gamma$, $\alpha_u+\alpha_v=w_v$;
	\item for other ranks $y_u,y_v$, we lower bound $\alpha_u$ and $\alpha_v$ by $w_v\cdot g(y_u,\gamma)$ and $w_v\cdot g(y_v,\tau)$ respectively.
\end{enumerate}

For the second part, the inequalities used in the proof of Lemma~\ref{lemma:v_gain} and \ref{lemma:u_gain} are tight only if $\beta,\theta$ are two step functions (refer to Figure~\ref{fig:theta_beta}). 
On the other hand, given these $\beta,\theta$, when $y_u\le \tau$ and $y_v\le \gamma$, we actually have $\alpha_u+\alpha_v=w_v$, which is strictly larger than our estimate $w_v\cdot (g(y_u,\gamma)+g(y_v,\tau))$.

With this observation, it is natural to expect an improved bound if we can retrieve this part of gain (even partially).
In this section, we prove an improved competitive ratio $0.6534$, using a refined lower bound for $\frac{1}{w_v}\cdot\expect{}{\alpha_u+\alpha_v}$ (compared to Lemma~\ref{lemma:bipartite_random}) as follows.

\begin{lemma}[Improved Bound]\label{lemma:bipartite_random_improved}
	For any pair of neighbors $u\in U$ and $v\in V$, we have
	\begin{align*}
		\frac{1}{w_v}\cdot\expect{}{\alpha_u + \alpha_v} \ge \min_{0\leq \gamma, \tau\leq1} \bigg\{ & (1-\tau)(1-\gamma) + (1-\tau)\int_0^\gamma g(x, \tau)dx \\
		& + \int_0^\tau \min_{\theta \leq \gamma}\Big\{ g(x,\theta)+\int_0^\theta g(y,x) dy +\int_\theta^\gamma g(y,\tau) dy \Big\} dx  \bigg\}.
	\end{align*}
\end{lemma}
\begin{proof}
	Let $\gamma$ and $\tau$ be defined as before, i.e., $\gamma = \beta(1)$ and $\tau = \min\{x:\theta(x)=1\}$.
	
	We divide $\frac{1}{w_v}\cdot \expect{}{\alpha_u + \alpha_v}$ into three parts, namely (1) when $y_u > \tau$ and $y_v > \gamma$; (2) when $y_u > \tau$ and $y_v < \gamma$; and (3) when $y_u < \tau$:
	\begin{align*}
		\frac{1}{w_v}\cdot \expect{}{\alpha_u + \alpha_v} = \quad &\frac{1}{w_v}\cdot \expect{}{(\alpha_u + \alpha_v)\cdot\onev(y_u > \tau,y_v > \gamma)} \\
		+ &\frac{1}{w_v}\cdot \expect{}{(\alpha_u + \alpha_v)\cdot \onev(y_u > \tau,y_v < \gamma)} \\
		+ &\frac{1}{w_v}\cdot \expect{}{(\alpha_u + \alpha_v)\cdot \onev(y_u < \tau)}.
	\end{align*}
	
	As shown in Lemma~\ref{lemma:corner_gain}, the first term is at least $(1-\tau)\cdot(1-\gamma)$, as we have $\alpha_u + \alpha_v = w_v$ for all $y_u > \tau$ and $y_v > \gamma$.
	Then we consider the second term, the expected gain of $\alpha_u + \alpha_v$ when $y_v < \gamma$ and $y_u > \tau$.
	For any $y_v < \gamma$, as we have shown in Lemma~\ref{lemma:v_gain}, $\alpha_v \geq w_v\cdot g(y_v,\beta^{-1}(y_v))$ for all $y_u > \tau$.
	Moreover, when $y_u < \beta^{-1}(y_v)$, we have $\alpha_u+\alpha_v = w_v$.
	Hence the second term can be lower bounded by
	\begin{equation*}
		\int_0^\gamma \Big( (1-\tau)\cdot g(y_v,\beta^{-1}(y_v)) + \max\{0, \beta^{-1}(y_v)-\tau\}\cdot \big( 1-g(y_v,\beta^{-1}(y_v)) \big) \Big) d y_v.
	\end{equation*}
	
	Now we consider the last term and fix a $y_u < \tau$.
	
	As we have shown in Lemma~\ref{lemma:u_gain}, for all $y_v\in[0,1]$, $\alpha_u \geq w_v\cdot g(y_u,\theta(y_u))$.
	
	Consider the case when $\theta(y_u) > \gamma$, then for $y_v \in (0, \gamma)$, $\alpha_v \geq w_v\cdot g(y_v, y_u)$;
	for $y_v \in (\gamma,\theta(y_u))$, $\alpha_u + \alpha_v = w_v$.
	Thus the expected gain of $\alpha_u + \alpha_v$ (taken over the randomness of $y_v$) can be lower bounded by
	\begin{equation*}
		w_v\cdot \Big( g(y_u,\theta(y_u)) + \int_0^\gamma g(y_v,y_u) dy_v + (\theta(y_u) - \gamma)\cdot(1-g(y_u,\theta(y_u))) \Big).
	\end{equation*} 
	
	As we have shown in Lemma~\ref{lemma:u_gain}, the partial derivative with respect to $\theta(y_u)$ is non-negative, thus for the purpose of lower bounding $\frac{1}{w_v}\cdot\expect{}{\alpha_u + \alpha_v}$, we can assume that $\theta(y_u) \leq \gamma$ for all $y_u < \tau$.
	
	Given that $\theta(y_u) \leq \gamma$, we have $\alpha_v \geq w_v\cdot g(y_v, y_u)$ when $y_v \in (0,,\theta(y_u))$; and $\alpha_v \geq w_v\cdot g(y_v, \beta^{-1}(y_v))$ when $y_v\in (\theta(y_u),\gamma)$.
	
	Hence the third term can be lower bounded by
	\begin{equation*}
		\int_0^\tau \Big( g(y_u,\theta(y_u)) + \int_0^{\theta(y_u)} g(y_v,y_u)d y_v + \int_{\theta(y_u)}^\gamma g(y_v,\beta^{-1}(y_v))d y_v \Big) d y_u 
	\end{equation*}
	
	Putting the three lower bounds together and taking the partial derivative with respect to $\beta^{-1}(y_v)$, for those $\beta^{-1}(y_v) > \tau$, we have a non-negative derivative as follows: 
	\begin{equation*}
		\frac{\partial g(y_v,\beta^{-1}(y_v))}{\partial \beta^{-1}(y_v)} + 1-g(y_v,\beta^{-1}(y_v))
		- (\beta^{-1}(y_v)-\tau)\cdot \frac{\partial g(y_v,\beta^{-1}(y_v))}{\partial \beta^{-1}(y_v)} \geq 0.
	\end{equation*}
	Thus for lower bounding $\frac{1}{w_v}\cdot \expect{}{\alpha_u + \alpha_v}$, we assume $\beta^{-1}(y_v) \leq \tau$ for all $y_v < \gamma$.
	Hence
	\begin{align*}
		\frac{1}{w_v}\cdot \expect{}{\alpha_u + \alpha_v} \ge & \min_{0\leq \gamma, \tau\leq 1}
		\Big\{  (1-\tau)(1-\gamma) + (1-\tau)\int_0^\gamma  g(y_v,\tau) d y_v \\
		& + \int_0^\tau \Big( g(y_u,\theta(y_u)) + \int_0^{\theta(y_u)} g(y_v,y_u)d y_v + \int_{\theta(y_u)}^\gamma g(y_v,\tau)d y_v \Big) d y_u 
		\Big\}.
	\end{align*}
	
	Taking the minimum over $\theta(y_u)$ concludes Lemma~\ref{lemma:bipartite_random_improved}.	
\end{proof}

Observe that for any $\theta\leq \gamma$, we have
\begin{equation*}
	g(x,\theta)+\int_0^\theta g(y,x) dy +\int_\theta^\gamma g(y,\tau) dy
	\geq g(x,\gamma) + \int_0^\gamma g(y,\tau) dy.
\end{equation*}
Thus the lower bound given by Lemma~\ref{lemma:bipartite_random_improved} is not worse than Lemma~\ref{lemma:bipartite_random}.

\begin{theorem}\label{th:improved}
	Fix $h(x) = \min\{1,\frac{1}{2}e^x\}$.
	For any pair of neighbors $u$ and $v$, and any fixed ranks of vertices in $U\cup V\setminus\{u,v\}$, we have
	$\frac{1}{w_v}\cdot \expect{y_u,y_v}{\alpha_u + \alpha_v} \ge 1-\frac{\ln{2}}{2} \approx 0.6534$.
\end{theorem}
\begin{proof}
	For $h(x) = \min\{1, \frac{1}{2}e^x\}$, we have $h'(x) = h(x)$ when $x<\ln(2)$, and $h'(x)=0$, $h(x) = 1$ when $x>\ln(2)$.
	
	Let $f(\tau,\gamma)$ be the expression on the RHS to be minimized in Lemma~\ref{lemma:bipartite_random_improved}.
	Using $g(x,y) = \frac{1}{2}(h(x+1-h(y)))$, we have
	\begin{align}
		f(\tau,\gamma) = & (1-\tau)(1-\gamma) + \frac{1-\tau}{2}\big( \gamma\cdot(1-h(\tau)) + \int_0^\gamma h(x)dx \big) \nonumber \\
		& + \frac{1}{2}\int_0^\tau \min_{\theta \leq \gamma}\Big\{ 1+\gamma+ h(x)-h(\theta)-\theta\cdot h(x)-(\gamma-\theta)\cdot h(\tau)
		+ \int_0^\gamma h(x) dx \Big\} dx \nonumber \\
		= & (1-\tau)(1-\gamma) + \frac{\gamma}{2}\cdot(1-h(\tau))+\frac{\tau}{2}+\frac{1}{2}\int_0^\gamma h(x) dx + \frac{1}{2}\int_0^\tau \min_{\theta \leq \gamma}\big\{ q(\tau, x,\theta) \big\} dx \label{eq:f(tau,gamma)},
	\end{align}
	where $q(\tau, x,\theta) := h(x)-h(\theta)-\theta\cdot h(x) + \theta\cdot h(\tau)$.
	Observe that
	\begin{equation*}
		\frac{\partial q(\tau, x,\theta)}{\partial \theta} = h(\tau) - h(x) -h'(\theta) \begin{cases}
			< 0 \qquad \text{when } \theta < \ln{2},\\
			\geq 0 \qquad \text{when } \theta \geq \ln{2}.
		\end{cases}
	\end{equation*}
	
	Thus we can lower bound $q(\tau,x,\theta)$ by (recall that $\theta \leq \gamma$ and $x<\tau$)
	\begin{equation*}
		q(\tau,x,\min\{\ln{2}, \gamma\}) \geq h(x)-h(\gamma)-\ln{2}\cdot h(x) + \ln{2}\cdot h(\tau).
	\end{equation*}
	
	Applying the lower bound on $q(\tau,x,\theta)$ in Equation~\eqref{eq:f(tau,gamma)}, we have
	\begin{align*}
		f(\tau,\gamma) \geq & (1-\tau)(1-\gamma) + \frac{\gamma}{2}\cdot(1-h(\tau))+\frac{\tau}{2}+\frac{1}{2}\int_0^\gamma h(x) dx \\
		& \qquad + \frac{1}{2}\int_0^\tau \Big( h(x)-h(\gamma)-\ln{2}\cdot h(x) + \ln{2}\cdot h(\tau) \Big) dx\\
		= & (1-\tau)(1-\gamma) + \frac{\gamma}{2}(1-h(\tau)) + \frac{\tau}{2}(1-h(\gamma)) + \frac{1}{2}\int_0^\gamma h(x) dx \\
		& \qquad + \frac{\ln{2}}{2}\tau\cdot h(\tau) + \frac{1-\ln{2}}{2}\int_0^\tau h(x) dx.
	\end{align*}
	
	In the following, we show that $f(\tau, \gamma) \geq 1-\frac{\ln{2}}{2} \approx 0.6534$ for all $\tau,\gamma\in[0,1]$, which (when combined with Lemma~\ref{lemma:bipartite_random_improved}) yields Theorem~\ref{th:improved}.
	
	First, observe that
	\begin{equation*}
		\frac{\partial f(\tau,\gamma)}{\partial \gamma} = -(1-\tau) -\frac{\tau}{2}\cdot h'(\gamma) + \frac{1}{2}(1-h(\tau)) + \frac{1}{2} h(\gamma) = \frac{1}{2}\Big((h(\gamma)-\tau\cdot h'(\gamma)) - (1+h(\tau)-2\tau)\Big).
	\end{equation*}
	which is non-decreasing in $\gamma$.
	
	Note that $1+h(\tau)-2\tau$ is strictly decreasing.
	Let $\tau^* \approx 0.3574$ be the solution for $1+h(\tau)-2\tau = 1$.
	Then we know that for $\tau \leq \tau^*$,
	\begin{equation*}
		\frac{\partial f(\tau,\gamma)}{\partial \gamma} \leq \frac{1}{2}\left( h(\gamma) - 1 \right) \leq 0
	\end{equation*}
	
	Thus,
	\begin{equation*}
		f(\tau,\gamma) \geq f(\tau, 1) = \frac{1}{2}(1-h(\tau)) + \frac{\ln{2}}{2} \tau\cdot h(\tau) + \frac{1}{2}\int_0^1 h(y) dy  + \frac{1-\ln{2}}{2}\int_0^\tau h(x)dx.
	\end{equation*}
	Recall that for $\tau < \tau^*$, $h'(\tau) = h(\tau) = \frac{1}{2}e^\tau$. Since
	\begin{equation*}
		\frac{\partial f(\tau,1)}{\partial \tau} = -\frac{1}{2}h(\tau)+\frac{\ln{2}}{2} h(\tau)+\frac{\ln{2}}{2}\tau\cdot h(\tau) + \frac{1-\ln{2}}{2} h(\tau) = \frac{\ln{2}}{2}\tau\cdot h(\tau) \geq 0,
	\end{equation*}
	we have (for $\tau < \tau^*$)
	\begin{equation*}
		f(\tau,\gamma) \geq f(\tau,1) \geq f(0,1) =  \frac{1}{2}(1-h(0)) + \frac{1}{2}\int_0^1 h(y) dy = 1-\frac{\ln{2}}{2} \approx 0.6534.
	\end{equation*}
	
	Now we consider $\tau > \tau^*$, in which case $1+h(\tau)-2\tau < 1$.
	
	Observe that $1+h(\tau)-2\tau > 1-\tau$ for all $\tau \in[0,1]$, we have
	\begin{equation*}
		\frac{\partial f(\tau,\gamma)}{\partial \gamma} \begin{cases}
			< 0 \qquad \text{ when } \gamma < \ln{2}, \\
			> 0 \qquad \text{ when } \gamma > \ln{2}.
		\end{cases}
	\end{equation*}
	
	Hence for $\tau > \tau^*$ we have
	\begin{equation*}
		f(\tau,\gamma) \geq f(\tau, \ln{2}) = 
		(1-\tau)(1-\ln{2}) +\frac{\ln{2}}{2}\cdot(1-h(\tau))+\frac{\ln{2}}{2}\tau\cdot h(\tau) + \frac{1}{4}+\frac{1-\ln{2}}{2}\int_0^\tau h(x)dx.
	\end{equation*}
	Taking derivative over $\tau$ on the RHS, we have
	\begin{equation*}
		\frac{\partial f(\tau, \ln{2})}{\partial \tau} = 
		-(1-\ln{2}) - \frac{\ln{2}}{2}\cdot h'(\tau) + \frac{\ln{2}}{2}\tau\cdot h'(\tau) + \frac{1}{2} h(\tau),
	\end{equation*}
	which is $\frac{1}{2}-(1-\ln{2}) > 0$ when $\tau > \ln{2}$.
	For $\tau \leq \ln{2}$, we have
	\begin{equation*}
		\frac{\partial f(\tau, \ln{2})}{\partial \tau} \begin{cases}
			< 0 \qquad \text{when } \tau < \tau_0,\\ 
			\geq 0 \qquad \text{when } \tau \geq \tau_0, 
		\end{cases}
	\end{equation*}
	where $\tau_0 \approx 0.564375$ is the solution of $\frac{\partial f(\tau, \ln{2})}{\partial \tau} = 0$.
	Thus for $\tau > \tau^*$ we have
	\begin{align*}
		f(\tau,\gamma) \geq f(\tau, \ln{2}) \geq f(\tau_0, \ln{2})  =
		& (1-\tau_0)(1-\ln{2}) +\frac{\ln{2}}{4}\cdot(2-e^{\tau_0} + \tau_0\cdot e^{\tau_0}) + \frac{1}{4} \\
		& +\frac{1-\ln{2}}{4}(e^{\tau_0}-1) \approx 0.6557 > 1-\frac{\ln{2}}{2}.
	\end{align*}
	
	Thus for all $\tau,\gamma\in[0,1]$, we have $f(\tau,\gamma)\geq 1-\frac{\ln{2}}{2}$, as claimed.
\end{proof}

\section{Conclusion}\label{sec:conclusion}

In this paper, we show that competitive ratios above $1-\frac{1}{e}$ can be obtained under the randomized primal-dual framework when equipped with a two dimensional gain sharing function.
The key of the analysis is to lower bound the expected combined gain of every pair of neighbors $(u,v)$, over the randomness of the rank $y_v$ of the offline vertex, and the arrival time $y_u$ of the online vertex.

Referring to Figure~\ref{fig:theta_beta}, it can be shown that the competitive ratio $F\geq \int_0^1 f(y_u) dy_u$, where
\begin{align*}
	f(y_u) := & \left(1-\theta(y_u)+\beta(y_u)\right)\cdot g(y_u, \theta(y_u)) + \theta(y_u) - \beta(y_u) \\
	& + \int_0^{\beta(y_u)} g(y_v, \beta^{-1}(y_v)) d y_v + \int_{\theta(y_u)}^1 g(y_v, \beta^{-1}(y_v)) d y_v.
\end{align*}

Note that here we assume $\beta^{-1}(y_v) = 1$ for all $y_v\geq \gamma$, and $g(x,1) = 0$ for all $x\in [0,1]$.

For every fixed $g$, there exist threshold functions $\theta$ and $\beta$ that minimize the integration.
Thus the main difficulty is to find a function $g$ such that the integration has a large lower bound for all functions $\theta$ and $\beta$ (which depend on the input instance).
We have shown that there exists a choice of $g$ such that the minimum is attained when $\theta$ and $\beta$ are step functions, based on which we can give a lower bound on the competitive ratio.

It is thus an interesting open problem to know how much the competitive ratio can be improved by (fixing an appropriate function $g$ and) giving a tighter lower bound for the integration.
We believe that it is possible to give a lower bound very close to (or even better than) the $0.696$ competitive ratio obtained for the unweighted case~\cite{stoc/MahdianY11}.

\section*{Acknowledgements}

The first author would like to thank Nikhil Devanur, Ankit Sharma, and Mohit Singh with whom he made an initial attempt to reproduce the results of Mahdian and Yan using the randomized primal-dual framework.

{
	\bibliographystyle{ACM-Reference-Format}
	\bibliography{matching}
}

\end{document}